\documentclass[letterpaper,english,preprintnumbers,amsmath,amssymb,superscriptaddress,nofootinbib,prl]{revtex4}
\usepackage{mathpazo}
\usepackage[T1]{fontenc}
\usepackage[latin9]{inputenc}
\usepackage{array}
\usepackage{verbatim}
\usepackage{multirow}
\usepackage{amsmath}
\usepackage{amsthm}
\usepackage{amssymb}
\usepackage{graphicx}
\usepackage{subfig}
\usepackage{setspace}
\usepackage{bbm}
\usepackage{pgf,tikz}
\usepackage{mathrsfs}
\usepackage{caption}
\usetikzlibrary{arrows}
\usepackage{hyperref}
\usetikzlibrary{automata,positioning}
\usetikzlibrary{decorations.pathmorphing,shapes}
\DeclareCaptionJustification{justified}{\leftskip=0pt \rightskip=0pt \parfillskip=0pt plus 1fil}
\usepackage{natbib}
\makeatletter

\pdfpageheight\paperheight
\pdfpagewidth\paperwidth

\@ifundefined{textcolor}{}
{%
 \definecolor{BLACK}{gray}{0}
 \definecolor{WHITE}{gray}{1}
 \definecolor{RED}{rgb}{1,0,0}
 \definecolor{GREEN}{rgb}{0,1,0}
 \definecolor{BLUE}{rgb}{0,0,1}
 \definecolor{CYAN}{cmyk}{1,0,0,0}
 \definecolor{MAGENTA}{cmyk}{0,1,0,0}
 \definecolor{YELLOW}{cmyk}{0,0,1,0}
}
  \theoremstyle{remark}
  \ifx\thechapter\undefined
    \newtheorem{rem}{\protect\remarkname}
  \else
    \newtheorem{rem}{\protect\remarkname}[chapter]
  \fi 
  \theoremstyle{remark}
  \newtheorem*{rem*}{\protect\remarkname}
  \theoremstyle{plain}
 \ifx\thechapter\undefined 
    \newtheorem{thm}{\protect\theoremname}
  \else
    \newtheorem{thm}{\protect\theoremname}[chapter]
  \fi
  \theoremstyle{plain}
 \ifx\thechapter\undefined 
    \newtheorem{lem}{\protect\lemmaname}
  \else
    \newtheorem{lem}{\protect\lemmaname}[chapter]
  \fi
  \theoremstyle{plain}
  \ifx\thechapter\undefined
    \newtheorem{conjecture}{\protect\conjecturename}
  \else
    \newtheorem{conjecture}{\protect\conjecturename}[chapter]
  \fi

\newtheorem{corollary}{Corollary}
\renewcommand{\vec}[1]{\mathbf{#1}}

\newcommand{\G}{G(n,\alpha,d)}
\makeatother

\usepackage{babel}
  \providecommand{\conjecturename}{Conjecture}
  \providecommand{\lemmaname}{Lemma}
  \providecommand{\remarkname}{Remark}
  \providecommand{\theoremname}{Theorem}
\begin{document}

\title{Directed Random Geometric Graphs}

\author{Jesse Michel}
\affiliation{Massachusetts Institute of Technology, Cambridge MA, 02139}
\author{Sushruth Reddy}
\affiliation{Massachusetts Institute of Technology, Cambridge MA, 02139}
\author{Rikhav Shah}
\affiliation{Massachusetts Institute of Technology, Cambridge MA, 02139}
\author{Sandeep Silwal}
\affiliation{Massachusetts Institute of Technology, Cambridge MA, 02139}
\author{Ramis Movassagh}
\email{R. Movassagh: ramis@us.ibm.com, Rest: mithack@mit.edu}
\selectlanguage{english}%
\affiliation{IBM Research, MIT-IBM AI Lab, Cambridge MA, 02142}
\date{\today}
\maketitle
Many real-world networks are intrinsically directed. Such networks include activation of genes, hyperlinks on the internet, and the network of followers on Twitter among many others \cite{gene_analysis, wikipedia_anlaysis, twitter_analysis}. The challenge, however, is to create a network model that has many of the properties of real-world networks such as powerlaw degree distributions and the small-world property \cite{highclustering, real_world_prop3}. To meet these challenges, we introduce the \textit{Directed} Random Geometric Graph (DRGG) model, which is an extension of the random geometric graph model. We prove that it is scale-free with respect to the indegree distribution, has binomial outdegree distribution, has a high clustering coefficient, has few edges and is likely small-world. These are some of the main features of aforementioned real world networks. We empirically observe that word association networks have many of the theoretical properties of the DRGG model.
\tableofcontents{}
\pagebreak
\section{\label{sec:Context-and-summary}Introduction}

The widespread availability of rich data on complex networks has spurred mathematical research into the properties of such networks. Researchers frequently use random graphs to model real networks \cite{power_in_out, power_law_in_random, distributed_algo, distributed_algo2, real_world_prop3, gene_analysis, twitter_analysis}.  There are many properties of real networks that appear in a diverse set of fields. The most prevalent of these properties are the following:

\textbf{Scale-free:} the distribution of the degrees of nodes is given by a power law, i.e. $P[\text{degree of }v=k]\sim k^{-\gamma}$ \cite{power_law_in_random}.  In general, power-laws frequently occur in systems where `rich' items become richer, though there are many potential explanations \cite{wikipedia_anlaysis,power_law_in_random}.  For networks, this means that nodes with high degree are more likely to attract connections from additional nodes added to the network.  This frequently occurs when the degree is a rough measure of popularity, for example, websites that are linked to by a lot of other websites are seen as credible, so are linked to by more websites \cite{www}.  In directed graphs, the indegree and outdegree distributions can be examined separately.  Networks can exhibit a power law indegree distribution but a skinny tail (e.g. Gaussian) outdegree distribution \cite{power_in_lowtail_out}.

\textbf{Few-edges:} Networks are usually quite sparse \cite{sparse}. Every node has a small expected degree that doesn't increase much as the number of nodes scales.  For example, if Twitter doubles its number of users, it's unlikely that the average user would increase the number of users they follow by any constant factor. Real networks tend to have $\widetilde{\mathcal{O}}(n)$ edges \footnote{$\widetilde{\mathcal{O}}(n)=\mathcal{O}(n\,\text{poly}(\log n))$}.

\textbf{Small-world:} Most nodes are connected via a relatively short number of hops in the graph \cite{newman,small_world}.  Specifically, the expected length of the shortest path between two randomly selected points is $O(\log n)$.  This property is related to the diameter, which is the longest shortest path in the graph.  Clearly, the diameter is an upper bound on the expected length of a shortest path.

\textbf{Clustering:} The \textit{clustering coefficient} is the average likelihood that given a node, a randomly selected pair of its neighbors will be connected via a single edge.  Intuitively, agents tend to make connections with those they are already `close' to.  We can write an expression for the clustering coefficient,
$$\frac{1}{n}\sum_uP[v\leadsto w|u\leadsto v,u\leadsto w],$$
where $v \leadsto w$ is the event that there is a directed edge from $v$ to $w$.
This paper finds a class of random graphs which satisfies all four of these key properties.

\section{\label{related_work}Previous work}
We briefly describe classes of random graphs that have attracted much mathematical attention. Perhaps the most prominent are Erd{\H o}s-Renyi random graphs, denoted $G(n, p)$. It is natural to extend this model to the directed case by including the directed edges $(v,u)$ and $(u,v)$ independently with probability $p$.   Erd{\H o}s-Renyi graphs are tractable and small-world but are not scale-free. 

Another well-known network with heavy-tailed power law node degree distributions is the preferential attachment model. While this model exhibits the desired power law distribution for degrees and is mathematically tractable, this model lacks certain key properties of some real world networks such as a realistic clustering coefficient, which measures the commonality of small communities \cite{pref_bad, clustering_prefattachment}.

Random geometric graphs (RGG) are another class of random graphs that have become more popular in recent years due to their simple formulation \cite{rggs}. In the simple RGG model, one uniformly distributes $n$ points over some space (often taken to be the unit $d-$dimensional cube or torus), and takes them to be the vertices of a random graph. One connects any pair of vertices with distance less than a fixed $R$. Intuitively nodes which are closer are more connected, which is a desired feature of some real world networks. Much is known about this class of graphs, such as the sizes of connected components and the minimum value of $R$ which results in a connected graph \cite{mthresh}. These types of random graphs are commonly used to model radio broadcasting towers and there has been much work on distributed algorithms on these models (see \cite{distributed_algo, distributed_algo2}). Some variations on the basic RGG model have also been studied. One variant, termed the $k$-nearest neighbor model, samples points as in the basic RGG model but only connects each point to its $k$ nearest neighbors \cite{RGGVariant1}. Another variant generalizes the RGG model by sampling points as before, but connecting nodes with a probability dependent on their distance \cite{RGGVariant2}.

\section{\label{model}Directed Random Geometric Graph (DRGG) Model}
 Consider $n$ points (nodes/vertices) uniformly distributed on a $d$-dimensional cube $[0,L]^d$. To make the model fully translation-invariant, we impose periodic boundary conditions, that is, for any $i$, $$(x_1,\cdots,x_{i-1},x_i+L,x_{i+1},\cdots,x_d)=(x_1,\cdots,x_{i-1},x_i,x_{i+1},\cdots,x_d).$$ 
This is equivalent considering points on a $d$-dimensional torus. For simplicity, below we take the unit cube (i.e., set $L=1$), but our results can be extended to a general $L$. Each point $v$ is assigned a radius $r_v$ distributed according to a Pareto distribution \cite{pareto},
\newcommand{\raddist}{f}
$$\raddist(r) = \begin{cases}
{\eta}/r^{\alpha} & r_0\leq r \leq 1/2\\
0 & \text{otherwise}
\end{cases},$$
where $\alpha>d+1$ is a parameter of the graph and $r_0$ is the minimum allowed radius, chosen based on $n$ so that the resulting graph is almost surely connected \cite{mthresh},
$$r_0=\left(\frac{\log n}{V_d\text{ }n}\right)^{1/d}.$$ 
where $V_d$ be the volume of the unit ball in $d$-dimensions (i.e. $v_1=2,v_2=\pi,v_3=4\pi/3$). $r \le 1/2$ ensures that balls of selected radii do not intersect themselves. $\eta$ is the normalizing factor that ensures $\int_{r_0}^{1/2}\raddist(r)dr=1$:
\[
    \eta=\frac{(\alpha-1)r_0^{\alpha-1}}{1-(2r_0)^{\alpha-1}}.
\]
We refer to the ball of radius $r_v$ centered at $v$ as ``$v$'s circle'' and denote the distance on the torus between $v$ and $u$ as $d(u,v)=d(v,u)$. The directed edges are added as follows. Add a directed edge from $u$ to $v$ (denoted by $u\leadsto v$) if $u$ is in $v$'s circle (i.e. $d(v,u)\le r_v$). Intuitively, $r_v$ is a measure of the popularity of $v$. If $v$ is unpopular, then only nodes close by will connect; conversely very popular nodes attract the attention of nodes that are far away. Lastly, since we want to model real-world networks, we can think of $d$ as being a small integer, such as $1 \le d \le 5$. See Figures \eqref{fig:drgg_fig} and \eqref{fig:drgg_fig2} for reference.
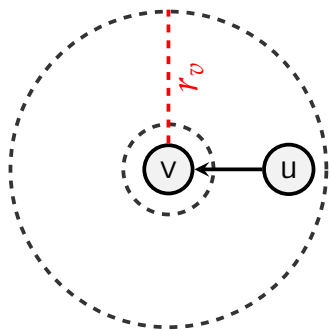
\begin{figure}[!htbp]
\begin{tikzpicture}[font=\sffamily]

        \tikzset{node style/.style={state, 
                                    minimum size=0.5cm,
                                    line width=0.5mm,
                                    fill=gray!10!white}}

        \node[node style] at (0, 0)         (v)     {\Large v};
        \node[node style] at (1.6, 0)     (u)     {\Large u};
        \node[circle, dashed, draw=black!80, thick, inner sep=0pt, minimum size=4.2cm, line width=0.5mm] (circ) {} (0,0) ;
        \node[circle, dashed, draw=black!80, thick, inner sep=0pt, minimum size=1.2cm, line width=0.5mm] (circ2) {} (1.6, 0) ;
        \draw[red,thick,dashed, line width=0.5mm] (v) -- (circ) node[midway, below, sloped] (EdgeAB) {\Large $r_v$};
        \draw[thick, line width=0.5mm, -stealth,decorate,decoration={snake,amplitude=0pt,pre length=2pt,post length=3pt}] (u) -- (v);
    \end{tikzpicture}
    \caption{\label{fig:drgg_fig}Illustration of $u\leadsto v$:  There is a directed edge from vertex $u$ to vertex $v$ since $u$ lies inside $v$'s circle of radius $r_v$. There is no directed edge from $v$ to $u$ since $v$ lies outside of $u$'s circle.}
\end{figure}
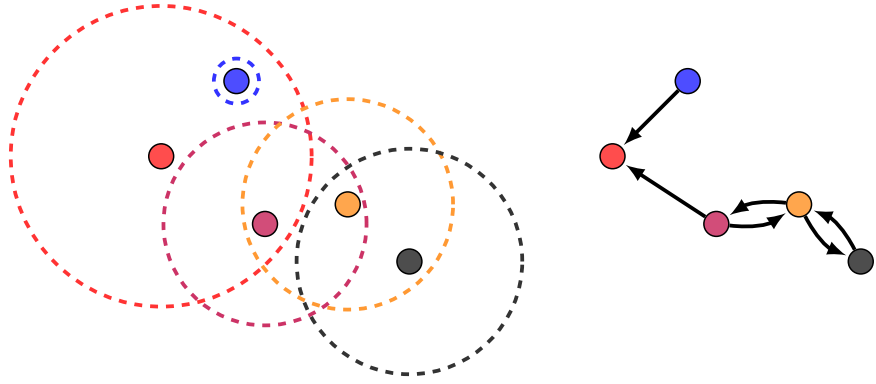
\begin{figure}[!htbp]
\begin{tikzpicture}[font=\sffamily]

        \tikzset{node style/.style={state, 
                                    minimum size=0.2cm,
                                    line width=0.2mm}}

        \node[node style, fill=red!70!] at (-10, 1)  (a) {};
        \node[node style, fill=blue!70!] at (-9, 2)  (c) {};
        \node[node style, fill=purple!70!] at (-8.62, 0.1)  (g) {};
        \node[node style, fill=orange!70!] at (-7.52, 0.36)  (i) {};
        \node[node style, fill=black!70!] at (-6.7, -.4)  (l) {};
        \node[circle, dashed, draw=red!80, thick, inner sep=0pt, minimum size=4cm, line width=0.5mm] (circ1) at  (-10, 1) {};
        \node[circle, dashed, draw=blue!80, thick, inner sep=0pt, minimum size=0.6cm, line width=0.5mm] (circ2) at (-9, 2) {};
        \node[circle, dashed, draw=purple!80, thick, inner sep=0pt, minimum size=2.7cm, line width=0.5mm] (circ4) at (-8.62, 0.1) {};
        \node[circle, dashed, draw=orange!80, thick, inner sep=0pt, minimum size=2.8cm, line width=0.5mm] (circ5) at (-7.52, 0.36) {};
        \node[circle, dashed, draw=black!80, thick, inner sep=0pt, minimum size=3cm, line width=0.5mm] (circ6) at (-6.7, -.4) {};
        
        \node[node style, fill=red!70!] at (-4, 1)  (a) {};
        \node[node style, fill=blue!70!] at (-3, 2)  (c) {};
        \node[node style, fill=purple!70!] at (-2.62, 0.1)  (g) {};
        \node[node style, fill=orange!70!] at (-1.52, 0.36)  (i) {};
        \node[node style, fill=black!70!] at (-0.7, -.4)  (l) {};

        \draw[every loop,
              auto=right,
              line width=0.5mm,
              >=latex]
            (c)     edge[right=20]            node {} (a)
            (g)     edge[right=20]            node {} (a)
            (g)     edge[bend right=20]            node {} (i)
            (i)     edge[bend right=20]            node {} (g) 
            (i)     edge[bend right=20]            node {} (l)
            (l)     edge[bend right=20]            node {} (i);  
    \end{tikzpicture}
    \caption{\label{fig:drgg_fig2} Left: Vertices with their respective circles plotted in the same color. Right: The corresponding DRGG.}
\end{figure}

We refer to this as the Directed Random Geometric Graph (DRGG) model, denoted by $G(n,\alpha,d)$. In summary the notation is:
\begin{itemize}
    \item $\G$ is a DRGG where $n$ is the number of nodes and $d$ is the dimension of the space.
    \item $\raddist(r)=\eta/r^{\alpha} $ probability density function for the selection of radii.
    \item $[r_0,1/2]$ is the support of $f(r)$.
\end{itemize}

\section{\label{math_results} Mathematical results}
For any class of random graphs, there are a variety of properties of theoretical and practical interests as outlined in the introduction. In this section, we prove that the indegree distribution of DRGGs is given by a power law with parameter $\frac{\alpha-1}d +1$ while the outdegree distribution is a binomial with parameter $\Theta(\frac{\log n}{n})$ where $n$ is the number of vertices. We also show that the total number of edges in DRGGs is $\Theta(n \log n)$ and show that the clustering coefficient of DRGGs approaches a constant as $n \rightarrow \infty$. Finally, we show how the standard geometric random graphs can be thought of as a graph limit of the DRGG model. We end with a conjecture about the diameter of DRGGs.
\subsection{Degree distributions}
As explained in Section \ref{sec:Context-and-summary}, motivated by real-world networks, we want a graph model whose indegree distribution follows a power law, and whose outdegree distribution decays exponentially.  Here we prove that indeed in DRGG, the indegrees follow a power law distribution, and outdegrees a binomial distribution.

\begin{thm}\label{thm1}
Let $\delta_{in}(k)$ and $\delta_{out}(k)$ be the density function for the indegree and outdegree of the random graph $\G$ respectively. Then, asymptotically as $n \to \infty$, 
\begin{equation}
\delta_{in}(k)\propto k^{-\frac{\alpha-1}{d}-1}
\end{equation}
for fixed $d$, $\alpha > d+1$, and $\frac{\alpha-1}d \ll k \ll n$ and 
\begin{equation}
\delta_{out}(k)= \dbinom{n-1}k z^k (1-z)^{n-1-k}
\end{equation}
where $z = \Theta(\log n/n).$
\end{thm}
\begin{proof}
We first prove the indegree case. Consider the probability that a vertex $v$ has indegree $k$. For this to happen, there must be exactly $k$ other vertices in $v$'s circle (of radius $r_v$). The probability of this is
\begin{equation}
P[\delta_{in}(v)=k] = \int_{r_0}^{1/2} P[r_v=r]\text{  } P[\delta_{in}(v)=k \mid r_v=r] \ dr.
\end{equation}
 Since the points are distributed uniformly on the plane, $P[\delta_{in}(v)=k \mid r_v=r]$ is the binomial random variable that describes the probability of exactly $k$ points lying within the volume of the $d$-dimensional sphere. This gives
\begin{equation}
P[\delta_{in}(v)=k] = \eta\dbinom{n}k\int_{r_0}^{1/2} \frac{1}{r^{\alpha}}(V_dr^d)^k(1-V_dr^d)^{n-k} \ dr.
\end{equation}
Making the substitution $u=V_dr^d$ gives us 
\begin{equation}\label{eq1}
         P[\delta_{in}(v)=k] = \frac{\eta}d\dbinom{n}{k}V_d^{\frac{\alpha-1}d}\int_{V_dr_0^d}^{V_d/2^d}u^{k-1- \frac{\alpha-1}d}(1-u)^{n-k} \ du.
\end{equation}
We let $\beta = \frac{\alpha-1}{d}$ for convenience and express the integrand in terms of the exponential function obtaining
\begin{equation*}
    P[\delta_{in}(v)=k] = \frac{\eta}d\dbinom{n}{k}V_d^{\beta}\int_{V_dm^d}^{V_d/2^d} \exp\left[n\left(\frac{k-\beta-1}{n}\ln u+(1-\frac{k}{n})\ln(1-u)\right)\right] \ du.
\end{equation*}
Now let $f(u) = \frac{k-\beta-1}{n}\ln u+(1-\frac{k}{n})\ln(1-u)$ so that the integrand is of the form $ e^{nf(u)}$. Treating $n$ as the large part and setting $f'(u)=0$, the saddle point is $u_{sp} = \frac{k-\beta-1}{n-\beta-1}$ and
\begin{equation*}
    f''(u_{sp}) = - \frac{(n-\beta-1)^3}{n(k-\beta-1)(n-k)}.
\end{equation*}
Note that for large $n$, $f''(u_{sp})$ approaches a negative constant. The steepest descent method gives (see \cite{laplace_method})
\begin{equation}
    \label{eq:steepest_result}
    P[\delta_{in}(v)=k] \approx \frac{\eta}d\dbinom{n}{k}V_d^{\beta}e^{nf(u_{sp})} \sqrt{\frac{2\pi}{n\mid f''(u_{sp})\mid}},
\end{equation}
which after inserting $f(u_{sp})$ and $f''(u_{sp})$ reads
\begin{equation}
\label{eq:inapprox}
    P[\delta_{in}(v)=k] \approx \frac{\eta}d\dbinom{n}{k}V_d^{\beta}\frac{\sqrt{2\pi}(k-\beta-1)^{k-\beta-1/2}(n-k)^{n-k+1/2}}{(n-\beta-1)^{n-\beta + 1/2}}.
\end{equation}
In fact, as $n \rightarrow \infty$, the two sides of Eq. \eqref{eq:steepest_result} become asymptotically equivalent \cite{laplace_proof}. As show in figure \ref{fig:degapprox}, the steepest descent result gives an excellent approximation. One can ignore the term $\frac{\eta}d$ since it do not depend on $k$ so they only contribute to the normalizing constant for $\delta_{in}.$ We now show that for large $n$, Eq. \eqref{eq:inapprox} gives rise to a power law in $k$. Using Stirling's approximation one obtains
$$\frac{1}{\sqrt{2\pi}\left(n-\beta-1\right)}\left(\frac{n}{k-\beta-1}\right)^{\beta+1} \leq \dbinom{n}{k}\frac{(k-\beta-1)^{k-\beta-1/2}(n-k)^{n-k+1/2}}{(n-\beta-1)^{n-\beta+1/2}} \leq \frac{1}{\sqrt{2\pi}\left(n-\beta-1\right)}\left(\frac{n}{k}\right)^{\beta+1}.$$
Ignoring factors that do not depend on $k$, we see that $P[\delta_{in}(v)=k]\propto k^{-\beta-1}$ as desired. 

We now analyze the outdegree case. Consider a vertex $v$. $v$ has outdegree $k$ if it lies in exactly $k$ circles of other vertices. Now consider any other vertex $u \ne v$. The probability that $v$ lies inside $u's$ circle is given by
$$z = \int_{r_0}^{1/2} P[v \text{ in } u's \text{ ball}|r_u=r]\,P[r_u=r]\,dr,
$$
where $ P[v \text{ in } u's \text{ ball}|r_u=r] = V_dr^d$ and $P[r_u=r]=\eta/r^{\alpha}$. Note that $z$ does not depend on $v$ or $u$. Since the radii of all circles are picked independently, it follows that the probability that the outdegree of vertex $v$ is $k$ is $\binom{n-1}kz^k(1-z)^{n-k}$. One can calculate $z$ exactly to be
$$z = \int_{r_0}^{1/2} \frac{\eta}{r^{\alpha}} \, V_d \, r^d \ dr  = \frac{\eta V_d}{d-\alpha+1} \left( \frac{1}{2^{d-\alpha+1}}-r_0^{d-\alpha+1}\right).$$
Defining $\alpha = \beta d + 1$ for $\beta > 1$ the above expression reads $$z = \frac{\beta V_d}{\beta - 1} \left(\frac{r_0^{\beta d}\, 2^{-d}-2^{-\beta d}r_0^d}{r_0^{\beta d}-2^{-\beta d}} \right).$$
Using Taylor expansion, and recalling that $\beta=\frac{\alpha-1}d$  and $1 \gg z$,  the equation above reduces to $$z = \left(\frac{\alpha-1}{\alpha-1-d}\right) \frac{\log n}n + \mathcal{O} \left(\frac{\log n}n\right)^{\beta-1}$$ and we are done.
\end{proof} 

\begin{figure}[hp]
\begin{minipage}{\linewidth}
  \centering
  \includegraphics[width=.4\linewidth]{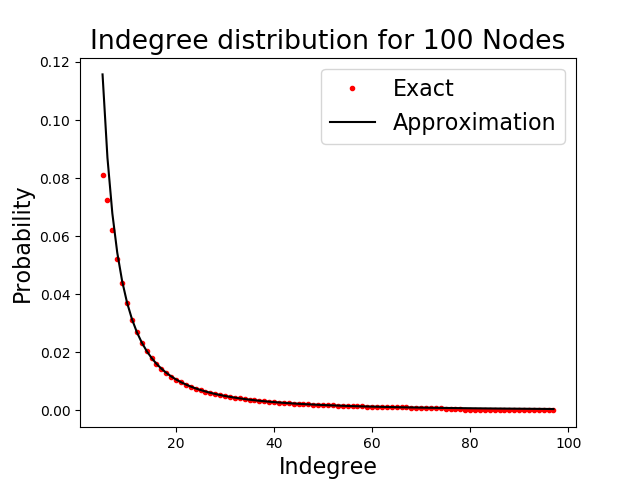}
  \includegraphics[width=.4\linewidth]{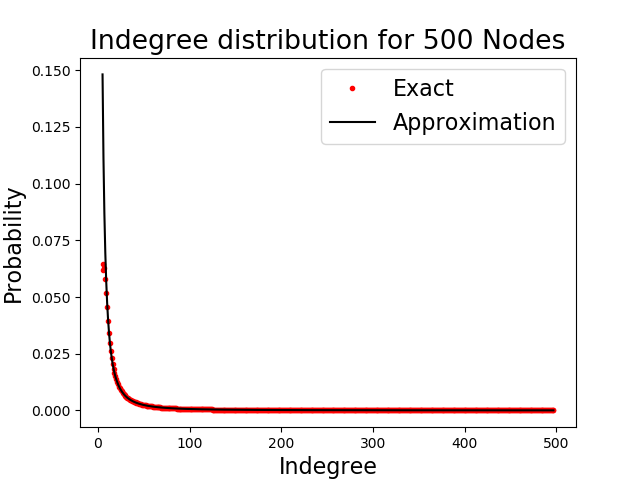}
\end{minipage}
\caption{\label{fig:degapprox}Illustration of the approximation of Eq. \eqref{eq:inapprox} for $100$, and  $500$  vertices for $\G$. The approximation is very accurate for both cases. For 100 nodes we see that the fit is tight except at the head and tail, which aligns with the fact that steepest descent makes a Gaussian approximation. We see a tighter fit for $n=500$, which aligns with the asymptotic nature of our approximation in \ref{thm1}. A heavy tail distribution can also be seen which numerically confirms the results proven in \ref{thm1}. }
\end{figure}

\begin{corollary}
The expected number of edges in $\G$ is $\Theta(n \log n).$
\end{corollary}
\begin{proof}
To find the total number of edges, it suffices to count all the outdegrees edges. Since there are $n-1$ nodes that $v$ can connect to with independent probabilities, Theorem \ref{thm1} implies that the expected outdegree of any vertex is $\Theta(\log n)$. Hence, the total number of edges is  $\Theta(n\log n)$.
\end{proof}

\subsection{Clustering Coefficient of $\G$ \label{sec:clustering}}
The clustering coefficient of $\G$ is discussed in this section. This metric is important because it measures the propensity of nodes with common neighbors to themselves be connected \cite{highclustering, directedclustering1, directedclustering2}. Hence, real world networks often exhibit large clustering coefficients. We being by showing that $\G$ is rich in triangles, namely, $\G$ has $\Theta(n \log^2(n))$ triangles. Triangles are important because they are a measure of connectivity and contribute to the clustering coefficient that will be discussed later in this section. In addition, counting the total number of triangles in a network is itself a well-studied problem \cite{triangles1}. In a directed graph, there are fundamentally two types of triangles: $\textbf{Type 1}$ and $\textbf{Type 2}$ triangles. These triangles are shown in Figure \eqref{fig:triangle}. 
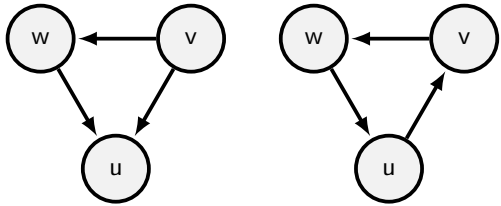
\begin{figure}[!htpb]
    \centering
    \subfloat{{ \begin{tikzpicture}[font=\sffamily]

        \tikzset{node style/.style={state, 
                                    minimum width=0.2cm,
                                    line width=0.5mm,
                                    fill=gray!10!white}}

        \node[node style] at (0, 0)     (u)     {w};
        \node[node style] at (2, 0)     (v)     {v};
        \node[node style] at (1, -1.73) (w) {u};

        \draw[every loop,
              auto=right,
              line width=0.5mm,
              >=latex]
            (u)     edge[right=20]            node {} (w)
            (v)     edge[right=20]            node {} (u)
            (v)     edge[right=20]            node {} (w);
    \end{tikzpicture}}}
    \qquad
    \subfloat{{\begin{tikzpicture}[font=\sffamily]

        \tikzset{node style/.style={state, 
                                    minimum width=0.2cm,
                                    line width=0.5mm,
                                    fill=gray!10!white}}

        \node[node style] at (0, 0)     (u)     {w};
        \node[node style] at (2, 0)     (v)     {v};
        \node[node style] at (1, -1.73) (w) {u};

        \draw[every loop,
              auto=right,
              line width=0.5mm,
              >=latex]
            (u)     edge[right=20]            node {} (w)
            (v)     edge[right=20]            node {} (u)
            (w)     edge[right=20]            node {} (v);
    \end{tikzpicture} }}
    \caption{\label{fig:triangle}A type $1$ triangle on the left and a type $2$ triangle on the right.}
\end{figure}
Due to the geometric nature of $\G$, it turns out that we do not need to account for both types of triangles. Say that vertices $u,v,w$ form a type $1$ or $2$ triangle if the edge relationships between $u,v$ and $w$ match Figure \eqref{fig:triangle}. Note that we say $v \leadsto w$ if there is a directed edge from $v$ to $w$.
\begin{lem}
\label{thm:Tri}
If vertices $u, v$ and $w$ form a type $2$ triangle in $\G$, then they also form a type $1$ triangle.
\end{lem}
\begin{proof}
We note that if $u,v$ and $w$ form a type $2$ triangle then $d(w,v) \le r_w$, $d(v,u) \le r_v, d(u,w) \le r_u$.
Without loss of generality, suppose that $r_v$ is the minimum radius: $r_v = \min(r_u, r_v, r_w).$ Then, we have $r_v \le r_u$ so it follows that 
$$ d(v,u) \le r_v \le r_u.$$
Hence, there must also be a directed edge from $v$ to $u$. That is, $u,v$ and $w$ also form a type $1$ triangle, as desired.
\end{proof}
Note that type $1$ triangles measure the following phenomenon in real networks: If two people follow the same person, what is the probability that one follows the other? Indeed, in real world networks, the number of triangles is `large', which indicates that models such as preferential attachment do not accurately model real world networks \cite{highclustering}. However, the total number of triangles of type $1$ in $\G$ is quite `large.' 
\begin{thm} \label{thm3}
The expected number of directed, acyclic (type 1) triangles in a graph, $|S|$, for $S = \{u, v, w : v \leadsto u, w \leadsto u, v \leadsto w\}$, is $\mathbb{E}[|S|] = \Theta(n\log^2(n))$ as $n \to \infty$ for fixed $\alpha, d$ under the constraint that $\alpha > 2d$.
\end{thm}
\begin{proof}We give an abridged version of the proof. Please see the appendix for full details. Throughout the proof, $P[\cdot]$ denotes the probability of an event while $p(\cdot)$ denotes the probability density of a random variable.
By linearity of expectation, $$\mathbb{E}[|S|] = n(n-1)(n-2) \, P(v \leadsto u, w \leadsto u, v \leadsto w)$$ for randomly selected vertices $u, v, w$. Without loss of generality, we may take $u$ to be located at the origin. Then condition on the locations of $v$ and $w$, denoted by $\vec{x_v}$ and $\vec{x_w}$, as well as the radii of circles centered at $u$ and $w$, denoted by $r_u$ and $r_w$. 

It is clear that the radii of circles corresponding to $u$ and $w$, along with the locations of $v$ and $w$ are independent. Additionally, edges are conditionally independent given these radii and their locations. The probability density of the points $x_v$ and $x_w$ is $1$ since $v, w$ are chosen uniformly over the unit torus: i.e. $p(\vec{x_v}) = p(\vec{x_w}) = 1$. Furthermore, $P[v \leadsto u|\vec{x_v}, r_u] = \mathbbm{1}[d(\vec{x_v}, 0) < r_u]$, as we draw an edge from $v$ to $u$ if and only if $\vec{x_v}$ is inside the circle of radius $r_u$ centered at $u$ (the origin). Similarly, $P[w \leadsto u|\vec{x_w}, r_u] = \mathbbm{1}[d(\vec{x_w}, 0) < r_u]$. Finally, we have that $P[v \leadsto w|\vec{x_w}, \vec{x_v}, r_w] = \mathbbm{1}[d(\vec{x_v}, \vec{x_w}) < r_w]$, as we draw an edge from $v$ to $w$ if and only if the distance between their coordinates is less than the radius of $w$'s circle. Note that $\mathbbm{1}[d(\vec{x_v}, \vec{x_w}) < r_w] \leq 1$ is always true. Using these ideas, we arrive at the following inequality
\[ P[v \leadsto u, w \leadsto u, v \leadsto w] \leq V_d^2 \eta \int_{r_0}^{\frac{1}{2}} r_u^{2d - \alpha} = V_d^2 \eta \left[\frac{r_u^{2d - \alpha + 1}}{2d - \alpha + 1}\right]^{\frac{1}{2}}_{r_0}.\]
As $n$ becomes large, for any fixed $d$ we have that $r_0$ goes to $0$. This then implies (for $\alpha > 2d$) that the dominant term in the above (and a valid, asymptotically tight upper bound) is:
\[ V_d^2 \eta \frac{r_0^{2d - \alpha + 1}}{\alpha - 2d - 1} = C_1 \frac{\log^2(n)}{n^2}\]
where $C_1$ is some constant independent of $n$. Using similar reasoning as above, we can get a lower bound on $P[v \leadsto u, w \leadsto u, v \leadsto w]$. Specifically, in the appendix, it is shown that
\begin{align*}
    &P[v \leadsto u, w \leadsto u, v \leadsto w] \ge V_d^2 \eta^2 \int_{r_0}^{\frac{1}{4}} dr_u r_u^{2d-\alpha} \left(\frac{(2r_u)^{1-\alpha} - (1/2)^{1-\alpha}}{\alpha - 1}\right) \\
    &= \frac{V_d^2 \eta^2}{\alpha - 1} \left[ \frac{2^{-\alpha}\text{ } r_0^{2d - 2\alpha - 2}}{\alpha -d - 1}\left( 1 - (4r_0)^{2(\alpha + 1)-2d} \right) - \frac{2^{\alpha-1}\text{ }r_0^{2d - \alpha + 1}}{\alpha -2d - 1} \left(1 - (4r_0)^{ \alpha - 1-2d}\right) \right].
\end{align*}
For fixed $d$, as $n$ goes to infinity, $r_0$ goes to $0$, so that the leading order term here is the one containing the factor $r_0^{2d-2\alpha -2}$ (as $\alpha > 2d > d - 1$). Thus, we have the lower bound
\[ P[v \leadsto u, w \leadsto u, v \leadsto w] \ge (C_2 - \epsilon) \frac{\log^2(n)}{n^2} \]
for any $\epsilon > 0$, for large enough $n$. This proves that, for fixed $\alpha, d$:
\[ C_1 \frac{\log^2(n)}{n^2} \geq P[v \leadsto u, w \leadsto u, v \leadsto w] \geq (C_2 - \epsilon) \frac{\log^2(n)}{n^2}. \]
Multiplying by $n(n-1)(n-2)$, we get:
\[ C_1 n\log^2(n) \geq \mathbb{E}[|S|] \geq (C_2 - \epsilon) n\log^2(n) \]
for any $\epsilon$, given that $n$ is sufficiently large. This proves that $\mathbb{E}[|S|] \in \Theta(n\log^2(n))$, as desired.
\end{proof}
We now turn to the clustering coefficient. For a directed graph, there exist many different possible definitions of the clustering coefficient \cite{directedclustering1,directedclustering2}. One natural choice is $\overline{C_{\text{in}}}$, which is defined as:
$$ \overline{C_{\text{in}}} = \frac{1}{|V|}\sum_{v \in V} c_v $$
where $c_u$ is the local clustering coefficient for vertex $u$ and is defined as:
\[ c_u = \frac{|\{v, w : v \leadsto u, w \leadsto u, v \leadsto w\}|}{d_u(d_u - 1)} \]
where $d_u$ is the degree of vertex $u$. In this section, we find the asymptotic expectation value of this clustering coefficient as $n \to \infty$, and show it approaches a constant (dependent on $d$ and $\alpha$). Note that the expectation is over all possible generated graphs. This is in contrast to other random graph models such as Erd{\H o}s-Renyi random graphs and preferential attachment graphs where the clustering coefficient are very small or approach $0$ as $n \rightarrow \infty$ \cite{clustering_prefattachment, clustering_erdos}. We then examine the case $\alpha \to \infty$, and show that it matches results from the random geometric graph model with radius $r_0$. This is in line with the correspondence outlined in the previous section, as in this limit, all edges become undirected and the radius concentrates at $r_0$, as in the case of random geometric graphs. Our main theorem is as follows:
\begin{thm} \label{thm5}
As $n \to \infty$ for fixed $d$, $\alpha$ satisfying $\alpha > 2d + 1$, the expectation of the clustering coefficient approaches a constant which is $\textbf{independent}$ of $n$ (the number of vertices). For $\textbf{odd}$ $d$, the constant can be calculated exactly:
\[ (\alpha - 1)^2 \frac{d!!}{(d-1)!!} \sum_{k=0}^{\frac{d-1}{2}} \frac{(-1)^k}{2k+1} \dbinom{\frac{d-1}{2}}{k} \left[\frac{1}{(\alpha - 1)^2 - d^2} - \frac{d}{2^{2k+1}(2k+d+1)}\left(\frac{1}{(\alpha - 1)^2 - (2k+d+1)^2}\right) \right].  \]
\end{thm}
\begin{proof}
We give an abridged proof. Please see the appendix for full details. By linearity of expectation (where the expectation is over all possible DRGG), $\mathbbm{E}[\overline{C}_\text{in}] = \mathbbm{E}[c_u]$ for a randomly selected vertex $u$. Moreover,
\[ \mathbbm{E}[c_u] = \mathbbm{E}\left[\frac{|\{v, w : v \leadsto u, w \leadsto u, v \leadsto w\}|}{d_u(d_u - 1)}\right] = \sum_{k} P(d_u = k) \frac{\mathbbm{E}[|\{v, w : v \leadsto u, w \leadsto u, v \leadsto w\}| | d_u = k]}{k(k - 1)} \]
from the law of iterated expectation. Furthermore, 
\[
\mathbbm{E}[|\{v, w : v \leadsto u, w \leadsto u, v \leadsto w\}| \, | d_u = k] = k(k - 1) P(v \leadsto w | v \leadsto u, w \leadsto u, d_u = k)
\]
by linearity of expectation. Then, we can substitute this expression into the sum given for $\mathbbm{E}[c_u]$:
\[ \mathbbm{E}[\overline{C}_\text{in}] = \sum_k P[d_u = k]P[v \leadsto w | v \leadsto u, w \leadsto u, d_u = k]. \]

After some computation, detailed in the Appendix, one obtains:
$$ \mathbbm{E}[\overline{C}_\text{in}] = \int dr_u p(r_u)\int dr_w p(r_w) \int_{||\vec{x_v}|| < r_u} d^d\vec{x_v}  \int_{||\vec{x_w}|| < r_u} d^d\vec{x_w} \frac{\mathbbm{1}[||\vec{x_v} - \vec{x_w}|| < r_w] }{V_d^2 r_u^{2d}}.$$

For odd $d$, we can compute this explicitly (work detailed in the Appendix):
\begin{equation}
\label{eq:evc_in}
\mathbbm{E}[\overline{C}_\text{in}] = \frac{d d!!}{(d-1)!!} \sum_{k=0}^{\frac{d-1}{2}} \frac{(-1)^k}{2k+1} \dbinom{\frac{d-1}{2}}{k} \left(\frac{1}{d} \frac{(\alpha-1)^2}{(\alpha - 1)^2 - d^2} - \frac{1}{2^{2k+1}(2k+d+1)} \frac{(\alpha-1)^2}{(\alpha - 1)^2 - (2k+d+1)^2} \right). 
\end{equation}

For even $d$, we cannot get nice closed forms since the density function of two randomly chosen points does not have a nice form like the one that exists for odd $d$. However, we can perform asymptotic analysis and show that the clustering coefficient also approaches a constant as $n \rightarrow \infty$ for even $d$ as well.
\end{proof}
\begin{rem}
Note that from the expression in Eq. \eqref{eq:evc_in}, it is clear that our value of $\mathbb{E}[\overline{C}_{in}]$ is a constant independent of $n$. This is different than other standard models such as the Erd{\H o}s-Renyi random graphs which are known to have low clustering coefficient \cite{clustering_erdos} and the preferential attachment model which have the property that the clustering coefficient approaches $0$ as $n \rightarrow \infty$ \cite{clustering_prefattachment}.
\end{rem}
It is interesting to take the limit as $\alpha \to \infty$ in Eq.\eqref{eq:evc_in} which gives:
$$ \lim_{\alpha \to \infty} \mathbbm{E}[\overline{C}_\text{in}] = \frac{d d!!}{(d-1)!!} \sum_{k=0}^{\frac{d-1}{2}} \frac{(-1)^k}{2k+1} \dbinom{\frac{d-1}{2}}{k} \left(\frac{1}{d} - \frac{1}{2^{2k+1}(2k+d+1)} \right).$$ This limit $\textbf{matches the known clustering coefficient derived by Dall and Christensen}$ of the standard RGG model \cite{rggs}. This is not surprising since taking $\alpha \rightarrow \infty$ in $\G$ results in the standard RGG model, as explored in the next section. The clustering coefficient for various odd dimensions along with their values as $\alpha \rightarrow \infty$ is shown in in Table \ref{table:limit_clustering} in the appendix.

\subsection{\label{undir_edge}Undirected edges and graph Limits}
In this section, we prove a lemma which shows that, given the existence of an edge, there is an asymptotically constant probability of an edge in the opposite direction. This shows that we have a positive fraction of what we have termed `undirected edges.'

\begin{lem} \label{lem_back}
For fixed $\alpha, d$, and for any randomly selected vertices $u, v$, as $n \to \infty$, we have that $P[u \leadsto v | v \leadsto u] \to \frac{2\beta - 2}{2\beta - 1}$ where $\beta = \frac{\alpha-1}d.$
\end{lem}
\begin{proof}
The conditional probability can be written as $\frac{P[u \leadsto v, v \leadsto u]}{P[v \leadsto u]}$. We now compute the numerator and denominator separately.
As before, we can without loss of generality situate $u$ at the origin. Then,
\[ P[u \leadsto v, v \leadsto u] =  \int_{r_0}^{\frac{1}{2}} dr_u
\int_{[0,1]^d} d^d\vec{x_v}
\int_{r_0}^{\frac{1}{2}} dr_v 
P[u \leadsto v, v \leadsto u | r_u, \vec{x_v}, r_v] \cdot p(r_u, \vec{x_v}, r_v). \]
Since $u, v$ share an undirected edge if and only if their separation is less than the radii of both of their circles, we have $$P[u \leadsto v, v \leadsto u | r_u, \vec{x_v}, r_v] = \mathbbm{1}[d(\vec{x_v}, 0) < r_u] \mathbbm{1}[d(\vec{x_v}, 0) < r_v]$$
Moreover,
$p(r_u, \vec{x_v}, r_v) = p(r_u) p(\vec{x_v}) p(r_v)$ as each of these $3$ quantities is chosen independently. Then, the integral is rewritten:
$$\int_{||\vec{x_v}|| < \frac{1}{2}} d^d\vec{x_v} \int_{||\vec{x_v}||}^{\frac{1}{2}} dr_u p(r_u)
\int_{||\vec{x_v}||}^{\frac{1}{2}} dr_v 
p(r_v) = \int_{||\vec{x_v}|| < \frac{1}{2}} d^d\vec{x_v}
P[r_u > ||\vec{x_v}||]P[r_v > ||\vec{x_v}||].$$
The integrand only depends on the norm of $\vec{x_v}$; we use spherical coordinates to write it as:
\[ S_{d-1} \int_0^{\frac{1}{2}} r^{d-1} P[r_u > r]P[r_v > r] dr \]
where $S_{d-1} = \frac{2\pi^{d/2}}{\Gamma(\frac{d}{2})}$. We can split this integral into portions that go from $0$ to $r_0$, where $P[r_u > r]$ and  $P[r_u > r]$ each evaluate to $1$, and a portion going from $r_0$ to $\frac{1}{2}$. By direct computation and for $b \geq r_0$, one obtains
\[ \int_{b}^{\frac{1}{2}} dr 
p(r) = \frac{\eta}{\alpha - 1} \left(\frac{1}{b^{\alpha - 1}} - 2^{\alpha - 1} \right), \]
and we have
\[ P[u \leadsto v, v \leadsto u] = S_{d-1} \left[ \frac{r_0^d}{d} + \frac{\eta^2}{(\alpha - 1)^2} \int_{r_0}^{\frac{1}{2}} r^{d-1} \left(\frac{1}{r^{\alpha -1}} - 2^{\alpha - 1} \right)^2 dr \right].\]
We are interested in the behavior of this integral as $n \to \infty$, or equivalently $r_0 \to 0.$ Therefore, for sufficiently large $n$,
$$\int_{r_0}^{1/2} r^{d-1} \left(\frac{1}{r^{\alpha-1}} - 2^{\alpha-1} \right) \ dr = \frac{r_0^{d - 2 \alpha + 2}}{2 \alpha - d -2} + \mathcal{O}(r_0^{d - \alpha + 1}) + C$$
for some constant $C$. Using the Taylor expansion of $\eta$ around $0$ and noting that $\alpha > d + 1$ (see \ref{model}), we get 
\begin{align*}
    P[u \leadsto v, v \leadsto u] &= S_{d-1} \left( \frac{r_0^d}{d} +  \frac{r_0^{d}}{2\alpha-d-2} \right) + \mathcal{O}(r_0^{d+\alpha- 1}).
\end{align*}

\noindent For the denominator of the conditional probability, assuming that $u$ is located at the origin, we compute
$$P[v \leadsto u] = \int_{r_0}^{\frac{1}{2}} dr_u p(r_u) P[v \leadsto u | r_u].$$
Note that $P[v \leadsto u | r_u] = P[||\vec{x_v}|| < r_u] = V_d r_u^d.$ Hence,
\[ P[v \leadsto u] = V_d \eta \int_{r_0}^{\frac{1}{2}} r^{-\alpha + d} \ dr =  V_d r_0^d \left(\frac{\alpha - 1}{\alpha - d - 1}\right) + \mathcal{O}(r_0^{\alpha -1}).\]
Dividing $P[u \leadsto v, v \leadsto u]$ by $P[v \leadsto u]$ and using the fact that $r_0 \rightarrow 0$, we get that
\[ \lim_{n \rightarrow \infty} P[u \leadsto v | v \leadsto u] = \frac{S_{d-1}}{V_d} \left( \frac{1}{d} + \frac{1}{2\alpha-d-2} \right) \frac{\alpha - d - 1}{\alpha - 1} = \frac{2\beta - 2}{2\beta - 1}, \]
since $\frac{S_{d-1}}{V_d} = d$ where $\beta = \frac{\alpha -1}d.$

\end{proof}
\begin{rem}
We can use Lemma \ref{lem_back} to understand the limiting behavior of $\G$ for fixed $n$ and $d$ and as $\alpha \rightarrow \infty$. In this case, the probability density function for the radii of the vertices converges to a delta distribution at the minimum radius $r_0$. However, we can actually say something stronger. In the $\alpha \rightarrow \infty$ case, $\G$ actually converges to a $\textbf{standard undirected random geometric graph}$ with fixed radius $r_0$. Note that $r_0$ is the sharp connectivity threshold for undirected random geometric graphs with fixed radius. To show this, note that the asymptotics we arrived at $P[u \leadsto v | v \leadsto u]$ are also valid for fixed $n$ and $d$ and $\alpha \rightarrow \infty$ since the $r_0$ term still dominates. Thus,
$$\lim_{\alpha \rightarrow \infty} P[u \leadsto v | v \leadsto u] = \frac{S_{d-1}}{dV_d} = 1.$$ This proves the following corollary.
\end{rem}
\begin{corollary}
For fixed $n$ and $d$, $\lim_{\alpha \rightarrow \infty} \G$ converges to the standard random geometric graph model (RGG).
\end{corollary}

\subsection{\label{path_lengths}Diameter}
The diameter of a graph is defined as the longest path among the set of shortest paths over all pairs of vertices. In our case, we are only concerned with directed paths. Similar to the clustering coefficient, the diameter of a graph is a measure of connectivity. It is common for real-world networks to have a small diameter, as can be seen from the popular `six-degrees of separation' phenomenon. Based on numerical results, we conjecture that the diameter of the DRGG model is $\mathcal{O}(\log n)$. It appears that the `fat tail' property of the radii distribution contributes to the significant reduction in the diameter of the DRGG model. However, even though we are not able to prove our conjecture for the diameter, it is still possible to prove a related result which hints that the diameter is indeed $\mathcal{O}(\log n)$.

\begin{lem}
\label{lem_evpath}
Let $\alpha$ and $d$ be fixed in $\G$. Pick two vertices $u$ and $v$ uniformly at random in $\G$. Let $a_k$ denote the number of directed paths of length $k$ from $u$ to $v$. If $k \ge \frac{\log n}{\log \log n}$, then $\mathbb{E}[a_k] \ge 1.$
\end{lem}

\begin{proof}
First consider three vertices $w_1, w_2$ and $w_3$ chosen uniformly at random. If there is a directed edge from $w_1$ to $w_2$, then this does not affect the probability of an edge from $w_2$ to $w_3$. This is because once the locations of $w_1$ and $w_2$ are fixed, the location of $w_3$ is still uniform, and having an edge from $w_2$ to $w_3$ only depends on the location of $w_3$. Thus, given a length $k$ directed path from $u$ to $v$, the edges of this path are independent of one another. Hence, if $a_k$ denotes the expected number of directed paths of length $k$ form $u$ to $v$, we have
\begin{equation}
    \label{eq_thm6}
    \mathbb{E}[a_k] = \dbinom{n-2}{k-1}z^k(k-1)!
\end{equation}
where $z$ is the probability of an edge in $\G$ which is given by $\frac{C\log n}n$ (Theorem \ref{thm1}) where $C = \frac{\alpha-1}{\alpha-1-d}$. Taking $n$ large, shifting $k \rightarrow k+1$, and using $\binom{n}k \ge \frac{n^k}{k^k}$, we have 
$$\mathbb{E}[a_k]  \ge \sqrt{2 \pi k}\left( \frac{C\log n}{e} \right)^k \, \frac{C \log n}n.$$
Hence, it suffices to find a $k$ such that $(\frac{\log n}e)^k \ge n.$ This is equivalent to finding a $k$ such that $k \log \log n - k \ge \log n.$ Rearranging, we see that $k \ge \frac{\log n}{\log \log n}$ works for large $n$, as desired. Therefore, the expected number of paths of length $k = \frac{\log n}{\log \log n}$ is at least $$\mathbb{E}[a_k] \ge C\sqrt{2 \pi k} \log k \ge 1.$$ 
\end{proof}
\begin{rem}
In fact we can take the log of the expression in Eq. \eqref{eq_thm6}, it can be shown that $k = \Omega \left( \frac{\log n}{\log \log n} \right)$ is the threshold for the expected number of paths being asymptotically greater than $1$.
\end{rem}
The above lemma tells us that we can expect to find a short path between any two vertices. However, this result is still far from establishing bounds on the diameter or even the length of the shortest path between two vertices chosen uniformly at random. We end this section with the conjecture.
\begin{conjecture}
Let $\alpha$ and $d$ be fixed in $\G$. The length of the diameter of $\G$ is $\mathcal{O}(\log n)$.
\end{conjecture}

\section{\label{num_results}An Application to real world networks}

We tested our model on a variety of real world networks. Our code is available at \url{https://github.com/martinjm97/DRGG}. Interestingly, we empirically observed that networks created through word association resulted in networks that had binomial outdegree distribution and power law indegree distribution. An example of this is the University of South Florida Word Association Network. To create this network, researchers asked participants to write the first word that came to mind that was meaningfully related or strongly associated to words that were presented to them. Then a directed edge was drawn between the word said by the participant and the word that was presented to them. This network has approximately $10^4$ vertices and $7.2 \times 10^4$ edges. For more information about this network, see \cite{word_network}. In this section we investigate this network and see how its properties compare to that of the DRGG model.

\subsection{Degree Distribution}
We begin by exploring how the indegree and outdegree distributions for the word association networks compare to the predictions of DRGG.
The outdegree and indegree distributions of the network along with the best fits according to DRGG are shown in Figure \ref{fig:modelfits}. Note that for degree distributions, DRGG essentially has $\textbf{one free parameter},$ namely $\beta = \frac{\alpha-1}d.$ Therefore, we fit both the outdegree and indegree distributions using $\beta$. As shown in Figure \ref{fig:modelfits}, DRGG is a close fit, especially considering that there was only one free parameter to tune. We discovered that the value of $\beta = 7/3 \approx 2.33$ ($\alpha = 8, d = 3)$ resulted in the best fit.

\begin{figure}[!htbp]%
    \centering
    \subfloat[Outdegree Distribution]{{\includegraphics[width=8.5cm]{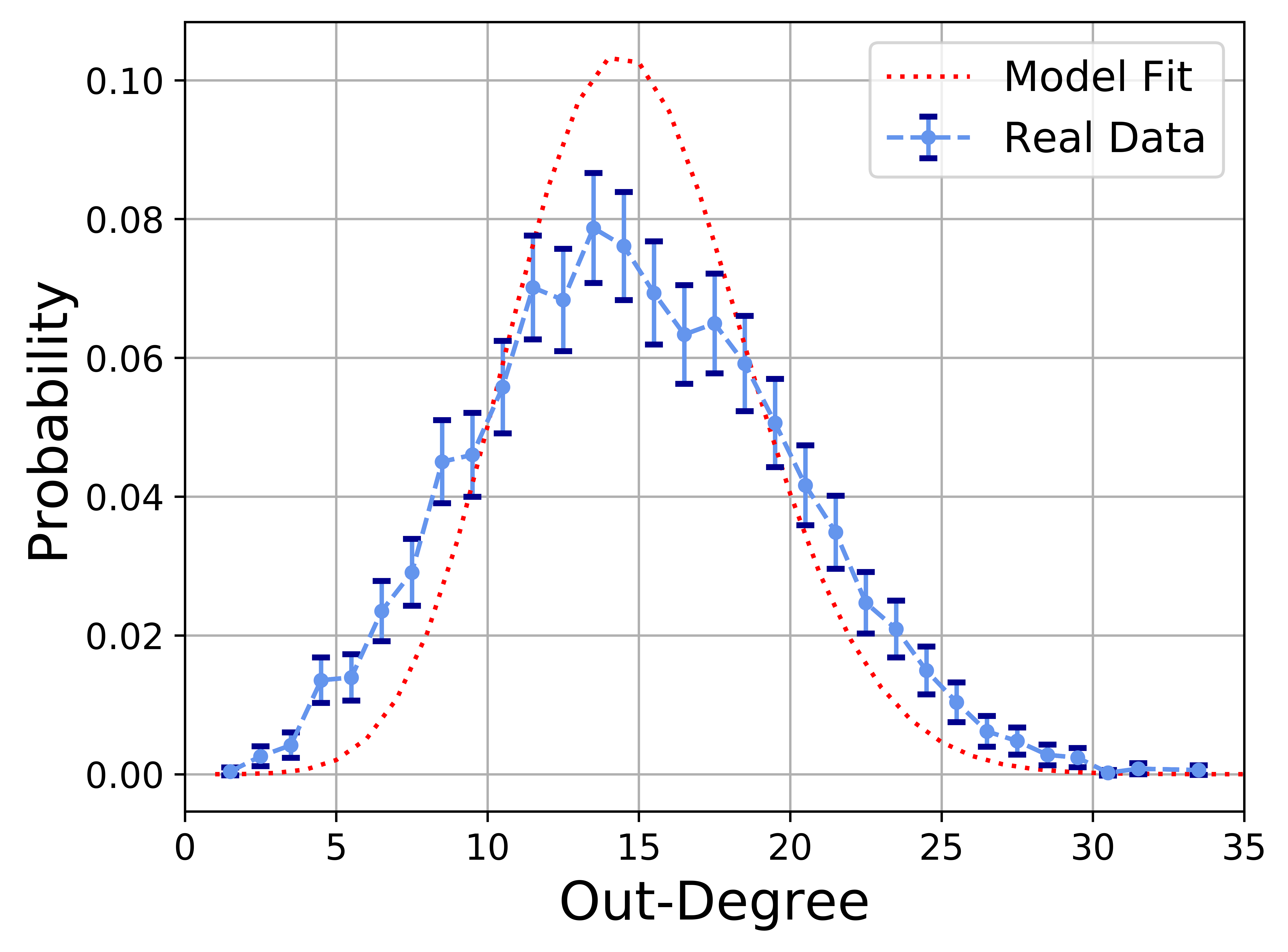} }}%
    \qquad
    \subfloat[Indegree Distribution]{{\includegraphics[width=8.5cm]{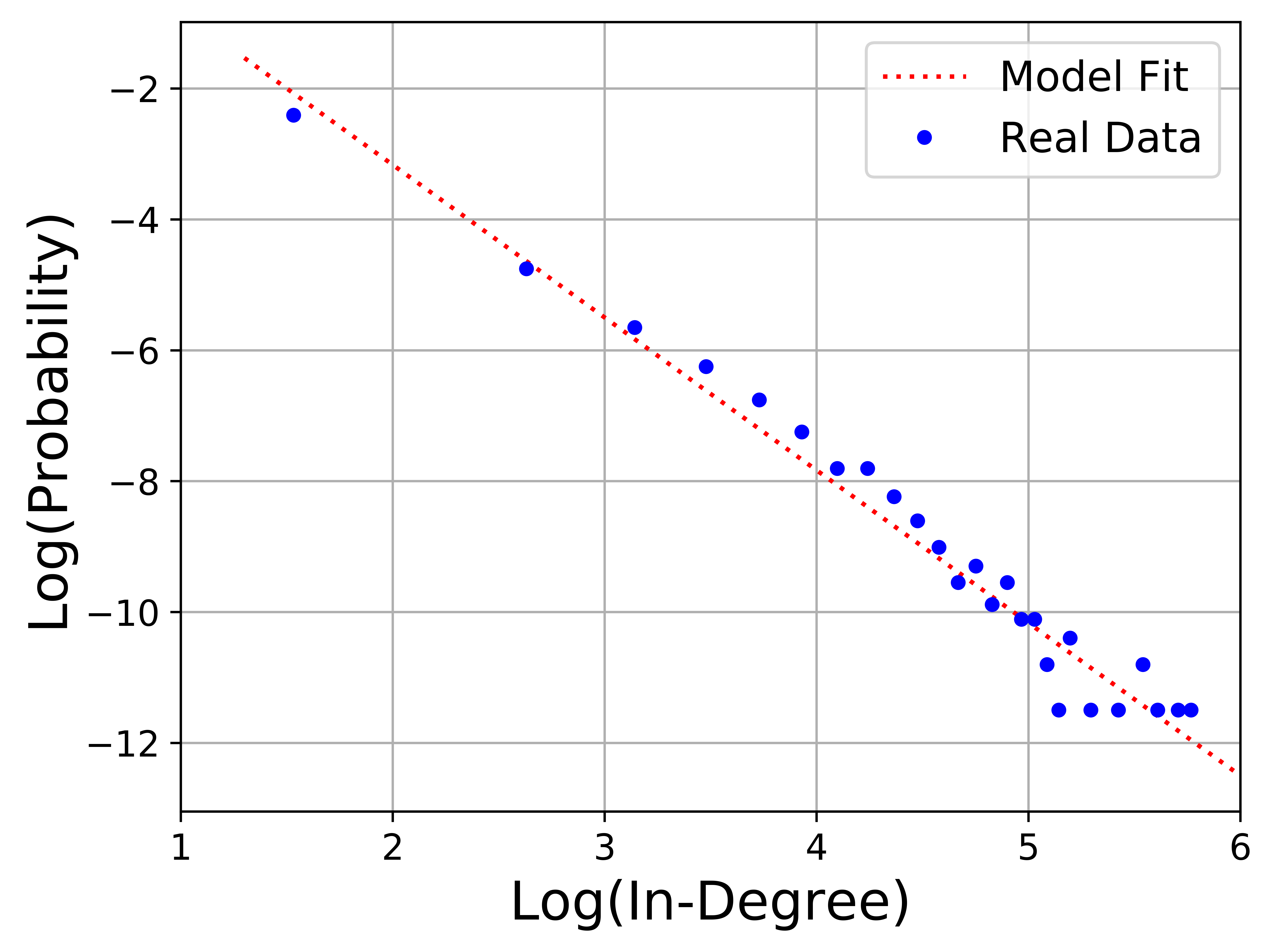} }}%
    \caption{Fit for outdegree and indegree distributions for the word association network. $\beta$ for both fits was $\beta = 7/3 \approx 2.33.$ The outdegree distribution on the left also has error bars of approximately two standards of deviation (we estimate one standard deviation of a bin by the square root of the number of items that fall in the bin).}%
    \label{fig:modelfits}%
\end{figure}

\subsection{Other Graph Statistics}
It is interesting to investigate how other statistics quantities compare to the values predicted by DRGG. In this section we use the model parameters of $\alpha = 8, d = 3$ to fit the degree distributions. 

We compared the average clustering coefficient (the clustering coefficient averaged over all nodes), the diameter, and the average path length. The average clustering coefficient gives information about how tightly-knit small communities are in the graph \cite{small_world}. Likewise, diameter and average path length provides insight into how closely the simulation models the real-world data. The comparison between our model predictions and actual real data values are shown in Table \ref{table:stat_table}. Note that the average clustering coefficient is noticeably lower in our model predictions than in the actual data. This suggests that the word association data set has a stronger clique behavior than our model. Overall, the empirical degree distributions and the empirical graph statistics do not quite match the predictions given by DRGG. This may be explained by the fact our results are asymptotic in $n$ and the network we studied only has $10^4$ vertices.

\begin{table}[!htbp]
\begin{center}
\begin{tabular}{|c|c|c|c|}
\hline
           & Avg. Clustering Coefficient & Diameter        & Avg. Path Length \\ \hline
Simulation & 0.512 $\pm$ $10^{-2}$       & 9.67 $\pm$ 0.94 & 5.149 $\pm$ 0.65 \\ \hline
Real Data  & 0.119                       & 7               & 4                \\ \hline
\end{tabular}
\end{center}

\caption{The results we averaged over 100 simulations of DRGG and were computed on the undirected version of the graph i.e., the directed edges were made by bidirectional (undirected). This was done for computational simplicity. Since the DRGG model is strongly connected, this should not have a large impact on the results. The same procedure was applied to the real data set. Finally, we only used the giant component in the real data set to calculate these statistics. The mean values for the three graph statistics along with two standard deviations of error are displayed.}
\label{table:stat_table}\end{table}
\subsection{Analysis of Hubs}
We analyze the words in the `hubs' of the word association network \cite{word_network}. By hubs, we specifically mean words that have a high indegree and represent the tail-end of the power law indegree distribution. The 20 words with the largest indegrees are shown in Table \ref{table:hubwords}. A lot of these hub words are $\textbf{emotional}$ words such as Love, Good, Bad, Pain, and Happy. In addition, the hubs also include words that are ubiquitous in everyday life, such as money, water, car, work, and people.

Further analysis was performed on the largest 50 hubs to understand their significance in relation to the remainder of the graph. One way of understanding the relationships between nodes in the graph is to look at their semantic similarity, or closeness in meaning. Several metrics have been proposed to quantify semantic similarity in words. We used two, based on the WordNet database and Word2Vec model. The WordNet database contains a hand-catalogued tree-like hierarchy of words. Given a word at node $n$, hypernyms (words with broader meanings) are located higher in the tree relative to $n$, while hyponyms (words with more specific meanings) are located lower in the tree. Two words can be judged to be similar in meaning if they are close together in the graph induced by these word relationships \cite{fellbaum2010wordnet}. One particular such measure of similarity is Wu-Palmer similarity:
\[ \text{Similarity score}(w_1, w_2) = \frac{2 \cdot \text{depth}(\text{least common ancestor}(w_1, w_2)}{\text{depth}(w_1) + \text{depth}(w_2)}. \]
The Wu-Palmer similarity always lies between $0$ and $1$, and is higher if the words are closer in semantic similarity \cite{DBLP:journals/corr/WuP94}. We found that the average similarity between a hub and its neighbors is $0.359$, the average similarity between an arbitrary node in the graph and its neighbors is $0.354$, and that the average similarity between two arbitrarily selected nodes in the graph is $0.241$. The average similarity between hubs is $0.263$. Thus, according to this similarity metric, hubs are on average slightly closer to their neighbors than other nodes. As expected, randomly chosen nodes have lower similarity as they may not be related, while nodes that are connected would be expected to have higher similarity. Hubs are also not very closely related: they are thus all common yet distinct words that serve as distinct ``sinks'' in the word graph.

To confirm this trend, we compared these results to those obtained from a different similarity metric. The Word2Vec model maps different words to continuous vector representations of a desired dimension. For more information, see the appendix. Similarity between vectors is then supposed to capture semantic similarity: indeed, relationships such as $\text{King} - \text{Man} + \text{Woman} = \text{Queen}$ seem to approximately hold between the vectors representing these respective words. We found that the average dot product similarity between a hub and its neighbors is $0.452$, the average similarity between an arbitrary node and its neighbors is $0.422$, and the average similarity between randomly selected nodes is $0.338$. The average similarity between different hubs is $0.415$. Thus, as before, hubs are on average slightly more similar to neighbors than arbitrarily chosen nodes, and much more similar than random words are to each other. Interestingly, this metric describes hubs as fairly similar to each other- this might be an artifact of this metric, as all hubs describe fairly common words that might appear together often in many text corpora. 

We also used a library named TextBlob to perform sentiment analysis to determine whether hubs expressed significantly different emotions than arbitrary nodes in the graph. For more information about TextBlob, please see the appendix.

Phrases can be constructed by concatenating words in the association network with their neighbors.  TextBlob was used to analyze these phrases.  Phrases that contain hub words have an average polarity of $0.0970$ and subjectivity of $0.553$, with variances of $0.0731$ and $0.0752$ respectively.
Phrases made from arbitrary graph nodes have an average polarity of $0.00511$ and subjectivity of $0.532$, with variances of $0.0982$ and $0.0611$ respectively. Thus, hubs tend to be slightly more ``positive'' than ``negative'', and slightly more ``subjective'' than ``objective'' compared to the average node in the graph, but these differences are on the order of the variances in these numbers. It thus seems that hubs are generally fairly neutral.

In summary, it seems that hubs are distinguished by the fact that they are short words that are easily memorable. This explains why they have high indegree: people have an easy time remembering them, regardless of their semantic content.

\begin{table}[!htbp]
\centering
\begin{tabular}{|c|c|c|c|c|}
\hline
Food: 324 & Money: 302 & Water: 276   & Car: 259    & Good: 255   \\ \hline
Bad: 229  & Work: 195  & House: 185   & School: 183 & Love: 181   \\ \hline
Man: 171  & Paper: 163 & Pain: 158    & Animal: 156 & People: 154 \\ \hline
Fun: 151  & Book: 149  & Clothes: 147 & Happy: 145  & Hard: 144   \\ \hline
\end{tabular}
\caption{Words with the largest indegrees in the word association network. The respective indegrees are shown next to the words.}
\label{table:hubwords}
\end{table}
\section{Conclusions and Future Work}
We have introduced a new model of random graphs, the Directed Random Geometric Graph (DRGG) model that has the property of being scale free in its indegree distribution, has few edges, and has a high clustering coefficient. Furthermore, we have displayed that this model can be applied to real world networks such as word association networks. Future work includes further theoretical investigation of the DRGG model, such as proving Conjecture 1 which states that the diameter of $G(n,\alpha, d)$ is $\mathcal{O}(\log n).$ 

A potential future application of DRGGs is to model power grid networks. Power grids are comprised of different types of nodes such as generators and transformers. Edges are directed because some nodes produce energy, others transfer energy, and yet others use up energy. Therefore, relationship among the nodes is asymmetrical. This matches the asymmetric indegrees and outdegrees of our model. Furthermore, similar to our model, connections are highly correlated with distance, since nodes that are far apart will be impractical to connect. However, a challenge for this analysis is the lack of directed data that is available, although there has been some recent progress to construct such a network by pulling data from multiple sources \cite{synthetic_grid}. We envision the analysis of directed power grid networks as a direction for future research with potential for broad applicability. 
In addition, we hope to see further applications of the DRGG model to other real world networks.

\section{Acknowledgements}
Jesse Michel, Sushruth Reddy, Rikhav Shah, and Sandeep Silwal would like to thank IBM Research for the opportunity to do an internship at the MIT-IBM AI lab in Cambridge MA in the winter of 2018, which was awarded to them for winning the 2017 HackMIT competition. They would also like to thank Ramis Movassagh for mentoring them during this period.

\section{Appendix}
\subsection{Proof of Theorem 2}
\begin{proof}
Throughout the proof, $P[\cdot]$ denotes a probability while $p(\cdot)$ denotes a probability density.
By linearity of expectation, we can write $\mathbb{E}[|S|] = n(n-1)(n-2) \, P(v \leadsto u, w \leadsto u, v \leadsto w)$ for randomly selected vertices $u, v, w$. Without loss of generality, as we are working on a torus, we may take $u$ to be located at $(0, 0, \cdots, 0)$. We can then condition on the locations of $v$ and $w$, which we denote $\vec{x_v}$ and $\vec{x_w}$, as well as the radii of circles centered at $u$ and $w$, which we denote $r_u$ and $r_w$, respectively, to obtain:
\[ P[v \leadsto u, w \leadsto u, v \leadsto w] = \int_{r_0}^{\frac{1}{2}} dr_u \int_{r_0}^{\frac{1}{2}} dr_w
\int_{[0,1]^d} d^d\vec{x_v}
\int_{[0,1]^d} d^d\vec{x_w}
\text{ }P[v \leadsto u, w \leadsto u, v \leadsto w | r_u, r_w, \vec{x_v}, \vec{x_w}] \text{ } p(r_u, r_w, \vec{x_v}, \vec{x_w}). \]
We now note that the radii of circles corresponding to $u$ and $w$, along with the locations of $v$ and $w$ are independent. Additionally, edges are conditionally independent given these radii and locations. Then, we may rewrite the above probability as:
\[\int_{r_0}^{\frac{1}{2}} dr_u \int_{r_0}^{\frac{1}{2}} dr_w
\int_{[0,1]^d} d^d\vec{x_v}
\int_{[0,1]^d} d^d\vec{x_w}
P[v \leadsto u|\vec{x_v}, r_u]\text{ }P[w \leadsto u | \vec{x_w}, r_u]\text{ }P[v \leadsto w | \vec{x_w}, \vec{x_v}, r_w]\text{ }p(r_u)\text{ }p(r_w)\text{ }p(\vec{x_v})\text{ }p(\vec{x_w}). \]
Note firstly that the probability density of the points $x_v$ and $x_w$ is $1$ since $v, w$ are chosen uniformly from the unit torus which means $p(\vec{x_v}) = p(\vec{x_w}) = 1$. Furthermore, $P[v \leadsto u|\vec{x_v}, r_u] = \mathbbm{1}[d(\vec{x_v}, 0) < r_u]$, as we draw an edge from $v$ to $u$ if and only if $\vec{x_v}$ is inside the circle of radius $r_u$ centered at $u$ (the origin). Similarly, $P[w \leadsto u|\vec{x_w}, r_u] = \mathbbm{1}[d(\vec{x_w}, 0) < r_u]$. Finally, we have that $P[v \leadsto w|\vec{x_w}, \vec{x_v}, r_w] = \mathbbm{1}[d(\vec{x_v}, \vec{x_w}) < r_w]$, as we draw an edge from $v$ to $w$ if and only if the distance between their coordinates is less than the radius of $w$'s circle. Substituting, we obtain:
\[ P[v \leadsto u, w \leadsto u, v \leadsto w]  = \int_{r_0}^{\frac{1}{2}} dr_u p(r_u) \int_{r_0}^{\frac{1}{2}} dr_w p(r_w) \int_{d(\vec{x_v}, 0) < r_u} d^d\vec{x_v} 
\int_{d(\vec{x_w}, 0) < r_u} d^d\vec{x_w} 
\mathbbm{1}[d(\vec{x_v}, \vec{x_w}) < r_w]. \]
We now upper and lower bound this expression in order to show that it is of order $\Theta(n\log^2(n))$.
We first prove an upper bound. 
\subsubsection{Upper Bound Calculation}
Note that $\mathbbm{1}[d(\vec{x_v}, \vec{x_w}) < r_w] \leq 1$ is always true, and so our probability satisfies:
\[ P[v \leadsto u, w \leadsto u, v \leadsto w] \leq \int_{r_0}^{\frac{1}{2}} dr_u p(r_u) \int_{r_0}^{\frac{1}{2}} dr_w p(r_w) \int_{d(\vec{x_v}, 0) < r_u} d^d\vec{x_v} 
\int_{d(\vec{x_w}, 0) < r_u} d^d\vec{x_w} \, 1. \]
Now, the last two integrals individually evaluate to the volume of the $d$-dimensional ball with radius $r_u$. Furthermore, the integral over $r_w$ evaluates to $1$, as $p(r_w)$ is a normalized probability density function. The remaining integral can be evaluated as follows:
\[ P[v \leadsto u, w \leadsto u, v \leadsto w] \leq V_d^2 \eta \int_{r_0}^{\frac{1}{2}} r_u^{2d - \alpha} = V_d^2 \eta \left[\frac{r_u^{2d - \alpha + 1}}{2d - \alpha + 1}\right]^{\frac{1}{2}}_{r_0}.\]
As $n$ becomes large, for any fixed $d$ we have that $r_0$ goes to $0$. This then implies (for $\alpha > 2d$) that the dominant term in the above (and a valid upper bound) is:
\[ V_d^2 \eta \frac{r_0^{2d - \alpha + 1}}{\alpha - 2d - 1} = C_1 \frac{\log^2(n)}{n^2}\]
where $C_1$ is some constant independent of $n$. 
\subsubsection{Lower Bound Calculation}
We now show a lower bound on the integral expression. We claim that:
\[P[v \leadsto u, w \leadsto u, v \leadsto w] \geq \int_{r_0}^{\frac{1}{2}} dr_u p(r_u) \int_{r_0}^{\frac{1}{2}} dr_w p(r_w) \int_{d(\vec{x_v}, 0) < r_u} d^d\vec{x_v} 
\int_{d(\vec{x_w}, 0) < r_u} d^d\vec{x_w}
\mathbbm{1}[r_w > 2 r_u] \]
This is because for choices of $r_w, r_u$ such that $r_w > 2 r_u$, we have $d(\vec{x_v}, \vec{x_w}) \leq d(\vec{x_v}, 0) + d(\vec{x_w}, 0) \leq 2r_u \leq r_w$ by the triangle inequality. Then, this integrand is always nonnegative and is identical to the original integrand when it is nonzero. This integral is then a lower bound on the original, as desired. Now, the integrals over $\vec{x_v}, \vec{x_w}$ can be done as before to give:
\[P[v \leadsto u, w \leadsto u, v \leadsto w] \geq V_d^2 \int_{r_0}^{\frac{1}{2}} dr_u r_u^{2d} p(r_u) \int_{r_0}^{\frac{1}{2}} dr_w p(r_w)
\mathbbm{1}[r_w > 2 r_u] \]
We can eliminate the indicator function by rewriting the bounds as:
\[  P[v \leadsto u, w \leadsto u, v \leadsto w] \ge V_d^2 \int_{r_0}^{\frac{1}{4}} dr_u r_u^{2d} p(r_u) \int_{2r_u}^{\frac{1}{2}} dr_w p(r_w) \]
Doing the integral over $r_w$, this becomes:
\begin{align*}
    &P[v \leadsto u, w \leadsto u, v \leadsto w] \ge V_d^2 \eta^2 \int_{r_0}^{\frac{1}{4}} dr_u r_u^{2d-\alpha} \left(\frac{(2r_u)^{1-\alpha} - (1/2)^{1-\alpha}}{\alpha - 1}\right) \\
    &= \frac{V_d^2 \eta^2}{\alpha - 1} \left[ \frac{2^{-\alpha}\text{ } r_0^{2d - 2\alpha - 2}}{\alpha -d - 1}\left( 1 - (4r_0)^{2(\alpha + 1)-2d} \right) - \frac{2^{\alpha-1}\text{ }r_0^{2d - \alpha + 1}}{\alpha -2d - 1} \left(1 - (4r_0)^{ \alpha - 1-2d}\right) \right].
\end{align*}
Note that for fixed $d$, as $n$ goes to infinity, $r_0$ goes to $0$, so that the leading order term here is the one containing the factor $r_0^{2d-2\alpha -2}$ (as $\alpha > 2d > d - 1$). Thus, we have the lower bound
\[ P[v \leadsto u, w \leadsto u, v \leadsto w] \ge (C_2 - \epsilon) \frac{\log^2(n)}{n^2} \]
for any $\epsilon > 0$, for large enough $n$. This then proves that, for fixed $\alpha, d$, that:
\[ C_1 \frac{\log^2(n)}{n^2} \geq P[v \leadsto u, w \leadsto u, v \leadsto w] \geq (C_2 - \epsilon) \frac{\log^2(n)}{n^2} \]
Multiplying by $n(n-1)(n-2)$, we get:
\[ C_1 n\log^2(n) \geq \mathbb{E}[|S|] \geq (C_2 - \epsilon) n\log^2(n) \]
for any $\epsilon$, given that $n$ is sufficiently large. This then proves that $\mathbb{E}[|S|] \in \Theta(n\log^2(n))$, as desired.
\end{proof}

\subsection{Proof of Theorem 3}

\subsubsection{Computing the expectation}
\begin{proof}
Now, note that $$P[v \leadsto w | v \leadsto u, w \leadsto u, d_u = k, r_u, r_w, \vec{x_v}, \vec{x_w}] = P[v \leadsto w | \vec{x_v}, \vec{x_w}, r_w] = \mathbbm{1}[||\vec{x_v} - \vec{x_w}|| < r_w]$$ (i.e. $v$'s lying in $w$'s circle is independent of all variables but for the positions of $v, w$ and $w$'s radius). Furthermore, $$p(r_u, r_w, \vec{x_v}, \vec{x_w}| v \leadsto u, w \leadsto u, d_u = k) = p(r_w)p(r_u, \vec{x_v}, \vec{x_w}| v \leadsto u, w \leadsto u, d_u = k),$$ from the conditional independence of $r_w$, as none of the other variables being considered involve edges pointing to $w$. Furthermore, this equals $p(r_w)p(r_u | v \leadsto u, w \leadsto u, d_u = k)p(\vec{x_v}, \vec{x_w}| r_u, v \leadsto u, w \leadsto u)$, by the chain rule of probability. Note that $p(r_u | v \leadsto u, w \leadsto u, d_u = k) = p(r_u | d_u = k)$ as $r_u$ is independent of the fact that two things lie within $u$'s circle given the indegree of $u$. Furthermore, note that 
$$p(\vec{x_v}, \vec{x_w}| r_u, v \leadsto u, w \leadsto u) = \frac{P[v \leadsto u, w \leadsto u, \vec{x_v}, \vec{x_w} | r_u]p(\vec{x_v}, \vec{x_w} | r_u)}{P[v \leadsto u, w \leadsto u | r_u]}$$ 
by Bayes' Rule. Note that $p(\vec{x_v}, \vec{x_w} | r_u) = p(\vec{x_v})p(\vec{x_w}) = 1$, that $P[v \leadsto u, w \leadsto u | r_u] = (V_dr_u^d)^2$, and that 
$$P[v \leadsto u, w \leadsto u | r_u, \vec{x_v}, \vec{x_w}] = \mathbbm{1}[||\vec{x_w}|| < r_u] \mathbbm{1}[||\vec{x_v}|| < r_u].$$

Now, we decompose the conditional probability above as:
\begin{align*}
    &P[v \leadsto w | v \leadsto u, w \leadsto u, d_u = k] \\
    &=\int dr_u \int dr_v \int d^d\vec{x_v} \int d^d\vec{x_w} P[v \leadsto w | v \leadsto u, w \leadsto u, d_u = k, r_u, r_w, \vec{x_v}, \vec{x_w}] p(r_u, r_w, \vec{x_v}, \vec{x_w}| v \leadsto u, w \leadsto u, d_u = k).
\end{align*}

Then, our integral becomes:
$$P[v \leadsto w | v \leadsto u, w \leadsto u, d_u = k] = \int dr_u p(r_u | d_u = k)\int dr_w p(r_w) \int_{||\vec{x_v}|| < r_u} d^d\vec{x_v}  \int_{||\vec{x_w}|| < r_u} d^d\vec{x_w} \frac{\mathbbm{1}[||\vec{x_v} - \vec{x_w}|| < r_w] }{V_d^2 r_u^{2d}}. $$

Substituting these results back into the original expression, we get the desired result in the main body of the paper.
\end{proof}
\subsubsection{Computing the integral expression for $\overline{C}_\text{in}$ for odd $d$}
\begin{proof}
Note that the inner $2$ integrals give the probability that $2$ points randomly chosen in a sphere of radius $r_u$ are less than $r_w$ apart. The answer to this question can be derived from a result that can be found in \cite{ball_distance_distribution}. Namely, for $\textbf{odd}$ $d$, the probability distribution function for the distance between two points being exactly $r$ apart in a ball of radius $R$ is:
$$P(r) = \frac{d r^{d-1}}{R^d}\frac{d!!}{(d-1)!!} \sum_{k=0}^{\frac{d-1}{2}} \frac{(-1)^k}{2k+1} \dbinom{\frac{d-1}{2}}{k} \left(1 - \left(\frac{r}{2R}\right)^{2k+1} \right). $$
(The even case is much harder to work with and does not result in a nice closed form. Thus, we will only work with the odd $d$ case).
The cumulative distribution function is then (for odd $d$):
$$D(r) = \frac{d d!!}{(d-1)!!} \sum_{k=0}^{\frac{d-1}{2}} \frac{(-1)^k}{2k+1} \dbinom{\frac{d-1}{2}}{k} \left(\frac{1}{d} \left( \frac{r}{R}\right)^d - \frac{1}{2^{2k+1}(2k+d+1)} \left(\frac{r}{R} \right)^{2k+d+1} \right).$$
Then, when we substitute back into the integral, we wish to compute integrals of the form:
$$\int dr_u p(r_u) \int dr_w p(r_w) \left(\frac{r_w}{r_u} \right)^m =  \frac{\eta^2 r_0^{2-2\alpha}}{(1 - \alpha)^2 -m^2} \, \left[ \left(1 -\left(\frac{1}{2 r_0}\right)^{m-\alpha + 1}\right)\left(1 -\left(\frac{1}{2 r_0}\right)^{-m-\alpha + 1}\right)\right]$$
where $m$ is an arbitrary integer. Note that  since $m \le (d-1)/2$, and $\alpha>2d+1$, we have that $\pm m-\alpha+1<0$. Then, as $n \to \infty$ (and, thus, as $r_0 \to 0$), the terms in parentheses both go to 1, so that the dominant term is proportional to $r_0^{2-2\alpha}$. We then get:
$$\int dr_u p(r_u) \int dr_w p(r_w) \left(\frac{r_w}{r_u} \right)^m \approx \frac{\eta^2}{(1-\alpha)^2 - m^2} r_0^{2-2\alpha} \to \frac{(\alpha-1)^2}{(\alpha - 1)^2 - m^2}. $$

Putting this all back together again, we get the final expression in the main body of the paper.

\begin{table}[!htbp]
\centering
\setlength{\tabcolsep}{15pt}
\def\arraystretch{3}
\begin{center}
\begin{tabular}{|c|c|c|}
\hline
$d$ & $\mathbb{E}[\overline{C}_{in}]$ &  $\lim_{\alpha \rightarrow \infty} \mathbb{E}[\overline{C}_{in}]$ \\ \hline
$1$ & $\frac{ (\alpha-1)^2}4 \left( \frac{4}{\alpha^2-2 \alpha} + \frac{1}{- \alpha^2-2\alpha + 1} \right)$ & $\frac{3}4$
      \\ \hline
$3$ & $\frac{3(\alpha-1)^2(5\alpha^4 -20\alpha^3 + 9\alpha^2 + 22\alpha -72)}{32(\alpha^2-2\alpha-8)(\alpha^2-2\alpha-5)(\alpha^2-2\alpha-3)}$ & $\frac{15}{32}$  \\ \hline
$5$ & $\frac{(\alpha-1)^2(159 \alpha^6 - 954\alpha^5 + 5364\alpha^4 -15096\alpha^3 - 73679\alpha^2 + 175006\alpha + 392040)}{512(\alpha^2-2\alpha-24)(\alpha^2-2\alpha-9)(\alpha^2-2\alpha-7)(\alpha^2-2\alpha-5)}$
  &  $\frac{159}{512}$ \\ \hline

\end{tabular}
\end{center}
\caption{Values of $\mathbb{E}[\overline{C}_{in}]$ for various odd dimensions. The limiting value as $\alpha \rightarrow \infty$ is also shown. We only have values for odd $d$ since we could not calculate a closed form expression for the probability of the distance between two randomly chosen points being exactly $r$ apart in a ball of radius $R$.}
\label{table:limit_clustering}
\end{table}
\end{proof}
\subsection{Analysis of Hubs}
The mapping utilized by Word2Vec is obtained by training a neural net. Specifically, the vectors are the solutions of an optimization problem which roughly attempts to maximize the dot products of vectors corresponding to words that are located close to each other in some text corpus. Ideally, words that are located close to each other often have similar meaning, and thus higher dot products, giving some indication of semantic similarity \cite{DBLP:journals/corr/MikolovSCCD13}. For the purposes of these experiments, we used a set of word vectors that were pre-trained on Google News articles. TextBlob uses a probabilistic model, tending to classify words as positive if they occur in many positive movie reviews, and negative if they occur in low-rated reviews. The library returns polarity and subjectivity values, which measure how negative/positive (on a $[-1, 1]$ scale) and objective/subjective (on a $[0, 1]$ scale) a given phrase is, respectively. 

\bibliography{Paper}

\begin{thebibliography}{33}
\providecommand{\natexlab}[1]{#1}
\providecommand{\url}[1]{\texttt{#1}}
\expandafter\ifx\csname urlstyle\endcsname\relax
  \providecommand{\doi}[1]{doi: #1}\else
  \providecommand{\doi}{doi: \begingroup \urlstyle{rm}\Url}\fi

\bibitem[Zhang and Horvath(2005)]{gene_analysis}
Bin Zhang and Steve Horvath.
\newblock A general framework for weighted gene coexpression network analysis.
\newblock In \emph{STATISTICAL APPLICATIONS IN GENETICS AND MOLECULAR BIOLOGY
  4: ARTICLE 17}, 2005.

\bibitem[Capocci et~al.(2006)Capocci, Servedio, Colaiori, Buriol, Donato,
  Leonardi, and Caldarelli]{wikipedia_anlaysis}
A.~Capocci, V.~D.~P. Servedio, F.~Colaiori, L.~S. Buriol, D.~Donato,
  S.~Leonardi, and G.~Caldarelli.
\newblock Preferential attachment in the growth of social networks: The
  internet encyclopedia wikipedia.
\newblock \emph{Phys. Rev. E}, 74:\penalty0 036116, Sep 2006.
\newblock \doi{10.1103/PhysRevE.74.036116}.
\newblock URL \url{https://link.aps.org/doi/10.1103/PhysRevE.74.036116}.

\bibitem[Ediger et~al.(2010)Ediger, Jiang, Riedy, Bader, and
  Corley]{twitter_analysis}
D.~Ediger, K.~Jiang, J.~Riedy, D.~A. Bader, and C.~Corley.
\newblock Massive social network analysis: Mining twitter for social good.
\newblock In \emph{2010 39th International Conference on Parallel Processing},
  pages 583--593, Sept 2010.
\newblock \doi{10.1109/ICPP.2010.66}.

\bibitem[{Kaiser} and {Hilgetag}(2004)]{highclustering}
M.~{Kaiser} and C.~C. {Hilgetag}.
\newblock {Spatial growth of real-world networks}.
\newblock \emph{\pre}, 69\penalty0 (3):\penalty0 036103, March 2004.
\newblock \doi{10.1103/PhysRevE.69.036103}.

\bibitem[Boccaletti et~al.(2006)Boccaletti, Latora, Moreno, Chavez, and
  Hwang]{real_world_prop3}
S.~Boccaletti, V.~Latora, Y.~Moreno, M.~Chavez, and D.-U. Hwang.
\newblock Complex networks: Structure and dynamics.
\newblock \emph{Physics Reports}, 424\penalty0 (4):\penalty0 175 -- 308, 2006.
\newblock ISSN 0370-1573.
\newblock \doi{https://doi.org/10.1016/j.physrep.2005.10.009}.
\newblock URL
  \url{http://www.sciencedirect.com/science/article/pii/S037015730500462X}.

\bibitem[{Tanimoto}(2009)]{power_in_out}
S.~{Tanimoto}.
\newblock {Power laws of the in-degree and out-degree distributions of complex
  networks}.
\newblock \emph{ArXiv e-prints}, December 2009.

\bibitem[Barabasi and Albert(1999)]{power_law_in_random}
Albert-Laszlo Barabasi and Reka Albert.
\newblock Emergence of scaling in random networks.
\newblock \emph{Science}, 286\penalty0 (5439):\penalty0 509--512, 1999.
\newblock ISSN 0036-8075.
\newblock \doi{10.1126/science.286.5439.509}.
\newblock URL \url{http://science.sciencemag.org/content/286/5439/509}.

\bibitem[Elsasser et~al.(1970)Elsasser, Gasieniec, and
  Sauerwald]{distributed_algo}
Robert Elsasser, Leszek Gasieniec, and Thomas Sauerwald.
\newblock On radio broadcasting in random geometric graphs.
\newblock \emph{International Symposium on Distributed Computing},
  5218:\penalty0 212--226, 01 1970.

\bibitem[Jia(2004)]{distributed_algo2}
Xingde Jia.
\newblock Wireless networks and random geometric graphs.
\newblock In \emph{7th International Symposium on Parallel Architectures,
  Algorithms and Networks, 2004. Proceedings.}, pages 575--579, May 2004.
\newblock \doi{10.1109/ISPAN.2004.1300540}.

\bibitem[Barabasi et~al.(2000)Barabasi, Albert, and Jeong]{www}
Albert-Laszlo Barabasi, Reka Albert, and Hawoong Jeong.
\newblock Scale-free characteristics of random networks: The topology of the
  world-wide web, 2000.

\bibitem[Behfar et~al.(2016)Behfar, Turkina, Cohendet, and
  Burger-Helmchen]{power_in_lowtail_out}
Stefan~Kambiz Behfar, Ekaterina Turkina, Patrick Cohendet, and Thierry
  Burger-Helmchen.
\newblock Directed networks different link formation mechanisms causing degree
  distribution distinction.
\newblock \emph{Physica A: Statistical Mechanics and its Applications},
  462\penalty0 (C):\penalty0 479--491, 2016.
\newblock URL
  \url{https://EconPapers.repec.org/RePEc:eee:phsmap:v:462:y:2016:i:c:p:479-491}.

\bibitem[{Huberman} et~al.(2008){Huberman}, {Romero}, and {Wu}]{sparse}
B.~A. {Huberman}, D.~M. {Romero}, and F.~{Wu}.
\newblock {Social networks that matter: Twitter under the microscope}.
\newblock \emph{ArXiv e-prints}, December 2008.

\bibitem[Newman(2010)]{newman}
Mark Newman.
\newblock \emph{Networks: An Introduction}.
\newblock Oxford University Press, Inc., New York, NY, USA, 2010.
\newblock ISBN 0199206651, 9780199206650.

\bibitem[Watts and H.~Strogatz(1998)]{small_world}
Duncan Watts and Steven H.~Strogatz.
\newblock Collective dynamics of small world networks.
\newblock \emph{Nature}, 393:\penalty0 440--2, 07 1998.

\bibitem[{Durak} et~al.(2012){Durak}, {Pinar}, {Kolda}, and
  {Seshadhri}]{pref_bad}
N.~{Durak}, A.~{Pinar}, T.~G. {Kolda}, and C.~{Seshadhri}.
\newblock {Degree Relations of Triangles in Real-world Networks and Models}.
\newblock \emph{ArXiv e-prints}, July 2012.

\bibitem[Bollobas(2003)]{clustering_prefattachment}
B.~Bollobas.
\newblock Mathematical results on scale-free random graphs.
\newblock In \emph{Handbook of Graphs and Networks}, pages 1--37. Wiley, 2003.

\bibitem[Dall and Christensen(2002)]{rggs}
Jesper Dall and Michael Christensen.
\newblock Random geometric graphs.
\newblock \emph{Phys. Rev. E}, 66:\penalty0 016121, Jul 2002.
\newblock \doi{10.1103/PhysRevE.66.016121}.
\newblock URL \url{https://link.aps.org/doi/10.1103/PhysRevE.66.016121}.

\bibitem[{Bradonjic} and {Perkins}(2013)]{mthresh}
M.~{Bradonjic} and W.~{Perkins}.
\newblock {On Sharp Thresholds in Random Geometric Graphs}.
\newblock \emph{ArXiv e-prints}, August 2013.

\bibitem[Balister et~al.(2008)Balister, Sarkar, and Bollobas]{RGGVariant1}
Paul Balister, Amites Sarkar, and Bela Bollobas.
\newblock \emph{Percolation, Connectivity, Coverage and Colouring of Random
  Geometric Graphs}, pages 117--142.
\newblock Springer Berlin Heidelberg, Berlin, Heidelberg, 2008.
\newblock ISBN 978-3-540-69395-6.
\newblock \doi{10.1007/978-3-540-69395-6_2}.
\newblock URL \url{https://doi.org/10.1007/978-3-540-69395-6_2}.

\bibitem[Dettmann and Georgiou(2016)]{RGGVariant2}
Carl~P. Dettmann and Orestis Georgiou.
\newblock Random geometric graphs with general connection functions.
\newblock \emph{Phys. Rev. E}, 93:\penalty0 032313, Mar 2016.
\newblock \doi{10.1103/PhysRevE.93.032313}.
\newblock URL \url{https://link.aps.org/doi/10.1103/PhysRevE.93.032313}.

\bibitem[Weisstein(2018)]{pareto}
Eric~W. Weisstein.
\newblock "pareto distribution." from mathworld--a wolfram web resource.
\newblock \url{http://mathworld.wolfram.com/ParetoDistribution.html}, 2018.

\bibitem[Chapling(2016)]{laplace_method}
Richard Chapling.
\newblock Asymptotic methods, April 2016.

\bibitem[Erdelyi(1956)]{laplace_proof}
A.~Erdelyi.
\newblock \emph{Asymptotic Expansions}.
\newblock Dover Books on Mathematics. Dover Publications, 1956.
\newblock ISBN 9780486603186.
\newblock URL \url{https://books.google.com/books?id=aedk-OHdmNYC}.

\bibitem[{Fagiolo}(2007)]{directedclustering1}
G.~{Fagiolo}.
\newblock {Clustering in complex directed networks}.
\newblock \emph{\pre}, 76\penalty0 (2):\penalty0 026107, August 2007.
\newblock \doi{10.1103/PhysRevE.76.026107}.

\bibitem[Tabak et~al.(2011)Tabak, Takami, Rocha, and
  Cajueiro]{directedclustering2}
Benjamin Tabak, Marcelo Takami, J.~Rocha, and Daniel Cajueiro.
\newblock Directed clustering coefficient as a measure of systemic risk in
  complex banking networks.
\newblock Working Papers Series 249, Central Bank of Brazil, Research
  Department, 2011.
\newblock URL \url{https://EconPapers.repec.org/RePEc:bcb:wpaper:249}.

\bibitem[Al~Hasan and Dave(2017)]{triangles1}
Mohammad Al~Hasan and Vachik~S. Dave.
\newblock Triangle counting in large networks: a review.
\newblock \emph{Wiley Interdisciplinary Reviews: Data Mining and Knowledge
  Discovery}, pages e1226--n/a, 2017.
\newblock ISSN 1942-4795.
\newblock \doi{10.1002/widm.1226}.
\newblock URL \url{http://dx.doi.org/10.1002/widm.1226}.
\newblock e1226.

\bibitem[Meghanathan(2015)]{clustering_erdos}
N.~Meghanathan.
\newblock A random network model with high clustering coefficient and variation
  in node degree.
\newblock In \emph{2015 8th International Conference on Control and Automation
  (CA)}, pages 54--57, Nov 2015.
\newblock \doi{10.1109/CA.2015.20}.

\bibitem[Nelson(1998)]{word_network}
Schreiber Nelson, McEvoy.
\newblock The university of south florida word association, rhyme, and word
  fragment norms.
\newblock \url{http://www.usf.edu/FreeAssociation/}, 1998.

\bibitem[Fellbaum(2010)]{fellbaum2010wordnet}
Christiane Fellbaum.
\newblock \emph{WordNet}.
\newblock Springer, 2010.

\bibitem[Wu and Palmer(1994)]{DBLP:journals/corr/WuP94}
Zhibiao Wu and Martha Palmer.
\newblock Verb semantics and lexical selection.
\newblock \emph{CoRR}, abs/cmp-lg/9406033, 1994.
\newblock URL \url{http://arxiv.org/abs/cmp-lg/9406033}.

\bibitem[Soltan et~al.(2017)Soltan, Loh, and Zussman]{synthetic_grid}
Saleh Soltan, Alexander Loh, and Gil Zussman.
\newblock A learning-based method for generating synthetic power grids.
\newblock 2017.

\bibitem[{Tu} and {Fischbach}(2002)]{ball_distance_distribution}
S.-J. {Tu} and E.~{Fischbach}.
\newblock {Random distance distribution for spherical objects: general theory
  and applications to physics}.
\newblock \emph{Journal of Physics A Mathematical General}, 35:\penalty0
  6557--6570, August 2002.
\newblock \doi{10.1088/0305-4470/35/31/303}.

\bibitem[Mikolov et~al.(2013)Mikolov, Sutskever, Chen, Corrado, and
  Dean]{DBLP:journals/corr/MikolovSCCD13}
Tomas Mikolov, Ilya Sutskever, Kai Chen, Greg Corrado, and Jeffrey Dean.
\newblock Distributed representations of words and phrases and their
  compositionality.
\newblock \emph{CoRR}, abs/1310.4546, 2013.
\newblock URL \url{http://arxiv.org/abs/1310.4546}.

\end{thebibliography}
\end{document}